%% file: main.tex
\documentclass{article}
\input{head}

\usepackage{todonotes}

\title{Quadratic Probing Insertions Are $\epsilon^{-(1+o(1))}$ Time}
\author{
  Yang Hu\thanks{Carnegie Mellon University. \texttt{yanghu2@andrew.cmu.edu}}
  \and
  William Kuszmaul\thanks{Carnegie Mellon University. \texttt{kuszmaul@cmu.edu}}
  \and
  Jingxun Liang\thanks{Carnegie Mellon University. \texttt{jingxunl@andrew.cmu.edu}}
  \and
  Stefan Walzer\thanks{Karlsruhe Institute of Technology. \texttt{stefan.walzer@kit.edu}}
  \and
  Huacheng Yu\thanks{Princeton University. \texttt{hy2@cs.princeton.edu}}
  \and
  Renfei Zhou\thanks{Carnegie Mellon University. \texttt{renfeiz@andrew.cmu.edu}}
}\date{}

\begin{document}

\maketitle

\begin{abstract}
First proposed in 1968, quadratic probing has stood for more than half a century as one of the simplest and most widely used hash-table designs in computer science. It is conjectured that, at load factor $1 - \epsilon$, the hash table achieves $O(\epsilon^{-1})$ expected insertion time. But even proving a bound of the form $f(\epsilon^{-1})$ for \emph{any} function $f$ has remained open. 

In this paper, we prove that the expected insertion time is $\epsilon^{-(1 + o(1))}$. This settles the complexity of the data structure up to sub-polynomial factors in $\epsilon^{-1}$.
\end{abstract}

\section{Introduction}
\label{sec:introduction}
\input{introduction}

\section{Preliminaries}
\label{sec:prelim}

\input{prelim}

\section{Poissonization and Positive Association}
\label{sec:poissonization}
\input{poissonization}

\section{A Tail Bound from the Covariance Properties}
\label{sec:covariance-tail-bound}
\input{covariance-tail-bound}

\section{Bounding Insertion Time}
\label{sec:tail-bounds}
\input{tail-bounds}

\section{Completing the Analysis}
\label{sec:expected-time}
\input{expected-time}

\section{Acknowledgments}

William Kuszmaul, Jingxun Liang, and Renfei Zhou were partially supported by NSF grant CNS-2504471, by NSF CAREER award CCF-2542165, and by a Jane Street Research Grant. Huacheng Yu was supported by NSF CAREER award CCF-2339942. Renfei Zhou was partially supported by the Jane Street Graduate Research Fellowship and the MongoDB PhD Fellowship. Part of this work was initiated by Dagstuhl Seminar 25191``Adaptive and Scalable Data Structures''. 

Parts of this paper were written with the assistance of AI tools. The concentration bound in Section \ref{sec:covariance-tail-bound} and the number-theoretic properties of the quadratic sequence in Section \ref{sec:tail-bounds} were both derived with the assistance of AI. Additionally, parts of the analysis in this paper were influenced by earlier discussions with AI about related problems. 

\bibliography{reference}

\end{document}

%% file: head.tex
\usepackage[T1]{fontenc}
\usepackage[margin=1in]{geometry}
\usepackage{amsfonts,amsmath,amsthm,amssymb,mathtools}  %
\usepackage{cases}  %

\mathtoolsset{centercolon}
\usepackage{xfrac,nicefrac}
\usepackage{mathdots}
\usepackage{mleftright}  %
\let\left\mleft
\let\right\mright

\usepackage{xspace}
\xspaceaddexceptions{]\}}  %
\usepackage{regexpatch}

\usepackage{bm,bbm,dsfont}  %
\usepackage{caption}
\usepackage[normalem]{ulem}
\usepackage{enumitem}

\usepackage{booktabs}
\usepackage{graphicx}
\usepackage{float}
\usepackage{subcaption}  %
\usepackage{tcolorbox}
\usepackage{tikz}
\usetikzlibrary{decorations.pathreplacing}
\usetikzlibrary{calc}
\usetikzlibrary{positioning}
\usetikzlibrary{arrows.meta}
\usetikzlibrary{math}
\usetikzlibrary{patterns}
\usetikzlibrary{trees}

\usetikzlibrary{trees}

\usepackage[linesnumbered,boxed,ruled,vlined]{algorithm2e}

\SetCommentSty{mycommfont}

\usepackage{thmtools,thm-restate}
\usepackage[colorlinks,citecolor=blue,linkcolor=blue,urlcolor=red]{hyperref}

\theoremstyle{plain}

\newtheorem{theorem}{Theorem}[section]  %
\newtheorem{lemma}[theorem]{Lemma}

\newtheorem{observation}[theorem]{Observation}
\newtheorem{proposition}[theorem]{Proposition}
\newtheorem{corollary}[theorem]{Corollary}

\newtheorem{claim}[theorem]{Claim}

\theoremstyle{definition}  %

\newenvironment{proofof}[1]{\begin{proof}[Proof of #1]}{\end{proof}}

\AtBeginDocument{%
}

\usepackage[capitalise]{cleveref}
\crefname{algocf}{Algorithm}{Algorithms}
\Crefname{algocf}{Algorithm}{Algorithms}
\crefname{claim}{Claim}{Claims}
\Crefname{claim}{Claim}{Claims}
\crefname{statement}{Statement}{Statements}
\Crefname{statement}{Statement}{Statements}
\crefname{observation}{Observation}{Observations}
\Crefname{observation}{Observation}{Observations}
\crefname{fact}{Fact}{Facts}
\Crefname{fact}{Fact}{Facts}
\crefname{property}{Property}{Properties}
\Crefname{property}{Property}{Properties}
\crefname{conjecture}{Conjecture}{Conjectures}
\Crefname{conjecture}{Conjecture}{Conjectures}
\crefname{hypothesis}{Hypothesis}{Hypotheses}
\Crefname{hypothesis}{Hypothesis}{Hypotheses}
\crefname{problem}{Problem}{Problems}
\Crefname{problem}{Problem}{Problems}
\crefname{openproblem}{Open Question}{Open Questions}
\Crefname{openproblem}{Open Question}{Open Questions}

\newfloat{Distribution}{htbp}{loa}
\crefname{Distribution}{Distribution}{Distributions}
\Crefname{Distribution}{Distribution}{Distributions}
\newfloat{Protocol}{htbp}{loa}
\crefname{Protocol}{Protocol}{Protocols}
\Crefname{Protocol}{Protocol}{Protocols}

\SetKwProg{Function}{Function}{:}{}

\SetKwProg{CodeBlock}{}{}{}
\SetKwProg{Repeat}{Repeat}{:}{}

\DeclarePairedDelimiter{\ceil}{\lceil}{\rceil}
\DeclarePairedDelimiter{\floor}{\lfloor}{\rfloor}

\DeclarePairedDelimiter{\bk}{(}{)}
\DeclarePairedDelimiter{\Bk}{[}{]}
\DeclarePairedDelimiter{\BK}{\{}{\}}

\DeclarePairedDelimiter{\abs}{\lvert}{\rvert}

\DeclarePairedDelimiterX\mysetbase[2]{\lbrace}{\rbrace}{#1\,\delimsize\vert\,#2}
\NewDocumentCommand{\myset}{sO{}m m}{%
  \IfBooleanTF{#1}%
    {\mysetbase*{#3}{#4}}%
    {\mysetbase[#2]{#3}{#4}}%
}

\DeclareMathOperator*{\E}{\mathbb{E}}
\DeclareMathOperator*{\Var}{Var}

\let\Pr\PrAux

\DeclareMathOperator*{\ind}{\mathds{1}}

\renewcommand{\tilde}{\widetilde}

\newcommand{\defeq}{\coloneqq}
\newcommand{\eps}{\varepsilon}

\renewcommand{\emptyset}{\varnothing}
\renewcommand{\epsilon}{\eps}

\newcommand{\defn}[1]{\emph{\boldmath\textbf{#1}}}

\usepackage{regexpatch}
\makeatletter
\xpatchcmd\thmt@restatable{%
\csname #2\@xa\endcsname\ifx\@nx#1\@nx\else[{#1}]\fi
}{%
\ifthmt@thisistheone
\csname #2\@xa\endcsname\ifx\@nx#1\@nx\else[{#1}]\fi
\else
\csname #2\@xa\endcsname[{Restated}]
\fi}{}{}
\makeatother

\usepackage{physics}

\usepackage{todonotes}

\newcommand{\Poisson}{\operatorname{Poisson}}
\newcommand{\Cov}{\operatorname{Cov}}

%% file: introduction.tex
The linear probing hash table, first introduced in 1954 by a group of researchers at IBM \cite{knuth1998art}, is one of the oldest and most basic data structures in computer science. To insert a key $x$ into a table of size $n$, linear probing computes a random hash $h(x)$, examines the slots
\[
    h(x),\ h(x)+1,\ h(x)+2,\ h(x)+3,\ldots \pmod n,
\]
and stores $x$ in the first empty slot from the sequence. This results in a hash table that is exceptionally simple to implement, and that exhibits good data locality in practice due to the consecutive nature of its memory accesses.

The main drawback of linear probing is that it scales poorly when the table is filled to a high load factor. Consider a table with load factor $1-\eps$, so that only an $\eps$ fraction of its slots remain empty. If successive probes behaved like independent random locations, one would expect to find an empty slot after $O(\eps^{-1})$ probes. This led Peterson in 1957 to conjecture that linear probing incurs $O(\epsilon^{-1})$ expected time per insertion \cite{peterson1957addressing}. However, in a landmark 1962 result, Knuth showed that this is not the case \cite{knuth1962notes}.\footnote{Knuth's note was never formally published, and the first published discovery of the result was in a 1966 paper by Konheim and Weiss \cite{konheim1966occupancy}.} This is due to a \emph{clustering effect}: occupied slots tend to form long consecutive runs, and these runs drive the expected number of probes needed to find an available slot up to $\Theta(\eps^{-2})$ \cite{knuth1998art,knuth1962notes,bender2021linear,braverman2024tight}.

In 1968, Maurer \cite{maurer1968programming} proposed quadratic probing as a simple way to retain much of linear probing's data locality while escaping its clustering behavior. Quadratic probing replaces the consecutive offsets by square offsets, examining slots from the sequence
\[
    h(x),\ h(x)+1^2,\ h(x)+2^2,\ h(x)+3^2,\ldots \pmod n.
\]
It is widely believed \cite{maurer1968programming,hopgood1972quadratic,richter2015seven,cormen2009introduction,drozdek1995data,bauer2015hashing,degreef2017hash,gries2014hashing,kesden2007hashing,sullivan2021hashing,weiss2000data} that quadratic probing avoids the clustering effects of linear probing, achieving an expected insertion time of $O(\epsilon^{-1})$. However, despite a great deal of interest \cite{maurer1968programming,hopgood1972quadratic,radke1970use,ecker1974period,batagelj1975quadratic,richter2015seven,kuszmaul2024analysis,guo2026simple,hu2026quadratic,zamir2026locality}, the proof of this conjecture has remained elusive.

Not only is it open to prove a bound of $O(\epsilon^{-1})$, it is open to prove any bound of the form $f(\epsilon^{-1})$ for \emph{any function} $f$. Kuszmaul and Xi made the first significant progress towards this goal in 2024, proving constant expected insertion time at a small positive load factor of roughly $9\%$ \cite{kuszmaul2024analysis}. Very recently, Guo, Pettie, and Wan gave a stronger analysis, implying $O(1)$ insertions up to load $37.61\%$ \cite{guo2026simple}. The only known bound to extend to higher load factors is $O(\eps^{-2}\log n)$ due to Zamir \cite{zamir2026locality} and depends on both $\epsilon$ and $n$. 

Thus, we find ourselves in the following situation. From a practical perspective, it is widely agreed that quadratic probing \emph{does} eliminate the clustering effects of linear probing \cite{cormen2009introduction,drozdek1995data,degreef2017hash,bauer2015hashing,gries2014hashing,kesden2007hashing,sullivan2021hashing,weiss2000data,richter2015seven,abseil2018raw}. On the other hand, from a theoretical perspective, we know almost nothing about its performance. Even proving an $O(1)$ expected-time bound for a table that is 50\% full remains open. 

\paragraph{This paper: a nearly optimal bound for quadratic probing.}
In this paper, we provide a nearly tight analysis of quadratic probing, establishing an $\epsilon^{-(1 + o(1))}$ bound on expected insertion time. Note that deletions are not considered.

\begin{restatable}{theorem}{thmmain}
  \label{thm:main}
  Consider a quadratic-probing hash table of prime size $n$. For any parameter $\eps$ with $1>\eps>1/\sqrt{\log n}$, the $((1-\eps)n + 1)$-th insertion takes expected time $\eps^{-(1+o(1))}$.
\end{restatable}

Thus, for any fixed constant $\epsilon \in (0, 1)$, we may conclude that quadratic probing supports $O(1)$ expected-time insertions. And, with $\epsilon^{-1}$ as an asymptotic parameter (satisfying $\eps \ge 1 / \sqrt{\log n}$), we may conclude that quadratic probing does indeed avoid the clustering effects of linear probing, at least up to a subpolynomial factor in $\epsilon^{-1}$.

The proof also offers a new perspective on why quadratic probing admits such a strong bound. 
Unlike the approaches of \cite{kuszmaul2024analysis,guo2026simple}, which extend to arbitrary ``fixed-offset probing'' hash tables\footnote{\cite{guo2026simple}, in their Section 4, tries to use the property of the quadratic probing sequence to improve their general bound, but achieves only a limited improvement. }, our analysis uses structure specific to quadratic probing hash tables (and must do so in order for a bound of $o(\epsilon^{-2})$ to be possible). Intuitively, the feature of quadratic probing that ends up being most important is the following: For any length-$k$ prefix $P_k$ of the quadratic probing sequence, and for any translation $t$, the intersection $P_k \cap (P_k + t)$ is guaranteed to be small (size at most $k^{o(1)}$, if $k \le O(\sqrt{n})$). Although this is not the only number-theoretic property of quadratic probing that we use in our analysis, it is (morally) the property that lets us provably avoid linear-probing-style clustering effects. 

Note that any analysis of quadratic probing must reason not just about the current insertion, but about the state of the hash table. We cannot afford, for example, to be the ``cartoon state'' where the $(1 - \epsilon)n$ elements all appear consecutively, and the $\epsilon n$ free slots also appear consecutively. 

A key observation in our analysis is that, although we do need to reason about the hash-table state, we can make use of a surprisingly weak characterization. In fact, although we require the final insertion we are analyzing to use quadratic probing, the only property we require of previous insertions is that they employed \emph{some} fixed-offset probing scheme. Moreover, although the resulting table structure could be quite sophisticated, our analysis ends up only making use of two facts about it: that, if we place $\Poisson((1 - \epsilon)n)$ keys into the table, and if we define $F_1, \ldots, F_m \in \{0, 1\}$ to be the indicators for which slots are free, then
\begin{enumerate}
    \item $(F_1, \ldots, F_m)$ has a translation-invariant distribution (by rotational symmetry), and
    \item the random variables $F_1, \ldots, F_m$ are positively associated. 
\end{enumerate}

At first glance, the positive association between free slots may seem like \emph{bad} news. We would greatly prefer if, instead, the free slots were negatively associated. This would mean that, if we probe a sequence of slots that are all occupied, it makes it more likely for subsequent probes to be free. And, indeed, if the slots were negatively associated, we could complete the entire analysis almost trivially using a Chernoff bound for negatively associated random variables. 

The second key observation in our analysis is that, although the association is not in the direction one might initially hope, having some structure is still far better than having no structure at all. In fact, positive association plays two central roles in our analysis. 

First, positive association guarantees that all pairwise covariances $\Cov(F_i, F_j)$ are nonnegative. This allows us to perform one of the key steps in our proof, which translates a bound of the form $\sum_{d \in [n]} \Cov(F_0, F_d) \le O(1)$ into a bound of the form $\sum_{d \in S} \Cov(F_0, F_d) \le O(1)$ for a specific set $S \subseteq [n]$ of our choosing. Such a translation is not possible without positive association, because the covariances that appear in the first sum but not the second could be highly negative. 

Second, although sums of positively associated random variables need not concentrate as sharply as sums of independent variables, one can prove concentration inequalities tailored to positive association, in which the loss is controlled by the sum of the pairwise covariances. In our setting, this is substantially stronger than applying Chebyshev's inequality directly to an arbitrary collection of dependent random variables. 

With these ideas in mind, the full analysis proceeds as follows. Section \ref{sec:poissonization} establishes the translational-invariance and positive-association properties that the analysis relies on; Section \ref{sec:covariances} then uses these properties to derive bound on the covariances along the quadratic-probing offset sequence; Section \ref{sec:covariance-tail-bound} derives a custom concentration inequality for bounding the probability that a sum of positively associated $\{0, 1\}$ random variables evaluates to 0; Section \ref{sec:tail-bounds} then derives the number-theoretic properties specific to the quadratic probing sequence that we need in our analysis; and finally, Section \ref{sec:expected-time} puts the pieces together in order to prove Theorem \ref{thm:main}.

\paragraph{Related work.} 
For most of the history of quadratic probing, the only known theoretical results \cite{maurer1968programming,hopgood1972quadratic,radke1970use,ecker1974period,batagelj1975quadratic} concerned the question of how many numbers in $\mathbb{Z}_n$ are hit by the sequence $(i^2)_{i = 0}^\infty$ (or, in some cases, related sequences such as $(\sum_{j = 0}^i j)_{i = 0}^\infty$). For prime $n$, the quadratic sequence ends up hitting only roughly half of $\mathbb{Z}_n$, so to formally ensure that every insertion succeeds, one must adapt the sequence so that the first $O(n)$ probes hit every slot.\footnote{In our paper, we assume that this is done in the most naive way possible where, after $n$ probes, the next $n$ probes simply scan through the table.} These early works \cite{maurer1968programming,hopgood1972quadratic,radke1970use,ecker1974period,batagelj1975quadratic} were not able to prove anything about expected insertion time. 

The first paper to prove a bound on expected insertion time was a 2024 paper by Kuszmaul and Xi, which used a witness-tree argument to obtain a constant time bound for the first $\approx 0.089 n$ insertions \cite{kuszmaul2024analysis}; this bound also applies to any fixed-offset open-addressing scheme. Subsequent work by Guo, Pettie, and Wan introduced a more sophisticated witness-forest analysis that applies to load factors up to $35.74\%$ for any fixed-offset scheme, and to load factors up to $37.61\%$ \cite{guo2026simple} for quadratic probing, specifically. Whether or not these types of witness tree/forest arguments can be extended to higher load factors remains an interesting open question. Prior to our work, the only known bound to extend to higher load factors was a result of Zamir \cite{zamir2026locality} which achieves a bound of $O(\epsilon^{-2} \log n)$ for any fixed-offset scheme.

Several recent works have considered ``quadratic-probing-like'' hash tables in different elements use different (randomized) offset sequences in order to simplify the analysis \cite{zamir2026locality, hu2026quadratic}. Hu et al.~\cite{hu2026quadratic} also extend their analysis to a version of quadratic probing in which every element uses the same fixed-offset probe sequence, but where that sequence is random across sequences with quadratic asymptotic growth; they show that, if the elements are kept in ``anti-Robin-Hood order'' (this requires a different insertion strategy than the standard one), then the resulting hash table is likely to support $O(\epsilon^{-1})$ amortized expected time insertions. 

Despite the lack of theoretical analysis, variations of quadratic probing have continued to be widely used in practice, including in widely used implementations by Google \cite{abseil2017abseil,google_dense_hash_map}, as well as in the standard libraries for both the Rust \cite{rust_hashmap} and Go \cite{go_builtin_map} programming languages.

The more-than 50-year quest to analyze quadratic probing comes in stark contrast to the analyses of other hash-table designs from the same era. For example, a close contemporary of quadratic probing is \emph{double hashing} \cite{balbine1968computational,bell1970linear}, which was also introduced in 1968, and which also continues to be an influential design today \cite{bronson2019open} (although perhaps less so than quadratic probing). Like quadratic probing, double hashing initially proved resistent to analysis. However, in a line of work spanning the 1970s and 80s \cite{guibas1976analysis,guibas1978analysis,lueker1988more,lueker1993more}, researchers were able to develop coupling techniques for relating the behavior of double hashing to that of the more straightforward random probing, and thereby obtained a tight analysis. Meanwhile, the task of obtaining almost any nontrivial analysis of quadratic probing continued to remain open.

%% file: prelim.tex
For a positive integer $m$, write $[m]=\{0,1,\ldots,m-1\}$. Slot indices are always taken modulo the table size. 
A \defn{quadratic probing hash table} stores a set $S \subseteq [U]$ of keys in an array of $n$ slots and supports insertions and membership queries. 

Throughout the paper, $h:[U]\to[n]$ denotes the hash function, which is assumed to map keys to independent uniformly random slots in $[n]$. When inserting a key $x \in [U]$, the algorithm examines slots
\begin{equation*}
    h(x),\ h(x)+1^2,\ h(x)+2^2,\ h(x)+3^2,\ldots 
\end{equation*}
in this order and stores $x$ in the first empty slot.\footnote{If the first $n$ probes find no empty slot, the algorithm examines the remaining unprobed offsets in an arbitrary fixed order; thus every insertion terminates within $2n$ probes so long as the hash table is not full.} We refer to this sequence of slots as the \defn{probe sequence} of $x$. A membership query examines the same sequence until it either finds $x$ or until it reaches an empty slot (meaning $x$ is not in the table).

For convenience, as in previous works, we restrict our attention to insertion-only workloads (no deletions). We remark, however, that one can also extend quadratic probing to support deletions, with up to $(1 - \epsilon)n$ elements present at a time, by simply marking elements as deleted (see, e.g., \cite{abseil2017abseil}), and rebuilding the hash table from scratch once every, say, $\epsilon n / 2$ operations. Such an approach naturally inherits the same $\epsilon^{-(1 + o(1))}$ time bound as the insertion-only case (plus the cost of performing the rebuilds, which amortizes to $\epsilon^{-(1 + o(1))}$ per operation). 

\paragraph{Distinct probes in a table of prime size.} Notice that, if $i^2 \equiv j^2 \pmod n$ then the $i$-th and $j$-th probes in the probe sequence are actually the same. The next observation says that, if $n$ is prime, then all of the first $n / 2$ probes are guaranteed to be distinct. 

\begin{observation}[Distinct probes]
  \label{obs:distinct-probes}
  Let $n$ be prime and let $1\le k\le n/2$. Then the offsets $0^2,1^2,\ldots,(k-1)^2$ are distinct modulo $n$. Equivalently, $b^2-a^2\not\equiv0\pmod n$ whenever $0\le a<b<k$.
\end{observation}

\begin{proof}
If $a^2\equiv b^2\pmod n$, then $n$ divides $(b-a)(b+a)$, so, $n$ being prime, $n$ divides $b-a$ or $b+a$. Neither is possible, since $0<b-a<k\le n/2$ and $0<b+a<2k\le n$.
\end{proof}

\paragraph{Harris's inequality and positively associated random variables.}
A finite family $X_1,\ldots,X_m$ of random variables is \defn{positively associated} if
\[
    \Cov\bk*{f(X_1,\ldots,X_m),g(X_1,\ldots,X_m)}\ge0
\]
for every pair of coordinatewise nondecreasing functions $f$ and $g$ for which the covariance is defined. 

The following classic result, due to Harris, says that independent random variables are positively associated.
\begin{theorem}[Harris's inequality \cite{harris1960lower}]
  \label{thm:harris}
  Let $Y_1,\ldots,Y_s$ be independent random variables. If $f$ and $g$ are both coordinatewise nondecreasing, then
  \[
      \Cov\bk*{f(Y_1,\ldots,Y_s),g(Y_1,\ldots,Y_s)}\ge0
  \]
  whenever the covariance is defined. The same conclusion holds when $f$ and $g$ are both coordinatewise nonincreasing. In other words, independent random variables are positively associated.
\end{theorem}

As an immediate corollary, the same result holds if the variables $Y_1,\ldots,Y_s$ are not themselves independent, but are instead all monotone functions of the same independent inputs.

\begin{corollary}[Monotone functions of independent variables are positively associated]
  \label{cor:monotone-association}
  Let $Y_1,\ldots,Y_s$ be independent random variables, and let $f_1,\ldots,f_m$ be real-valued functions that are either all coordinatewise nondecreasing or all coordinatewise nonincreasing. Then the variables
  \[
      X_i\defeq f_i(Y_1,\ldots,Y_s),\qquad i=1,\ldots,m,
  \]
  are positively associated.
\end{corollary}

\begin{proof}
We consider only the case where the $f_i$s are all coordinatewise nondecreasing; the nonincreasing case follows symmetrically. Let $F$ and $G$ be coordinatewise nondecreasing, and write $\tilde F(y)\defeq F\bk*{f_1(y),\ldots,f_m(y)}$ and $\tilde G(y)\defeq G\bk*{f_1(y),\ldots,f_m(y)}$, so that $F(X_1,\ldots,X_m)=\tilde F(Y_1,\ldots,Y_s)$ and $G(X_1,\ldots,X_m)=\tilde G(Y_1,\ldots,Y_s)$. Increasing a coordinate $y_j$ can only increase each $f_i(y)$, and can therefore only increase $\tilde F(y)$ and $\tilde G(y)$. Thus $\tilde F$ and $\tilde G$ are coordinatewise nondecreasing, and \cref{thm:harris} gives $\Cov(\tilde F,\tilde G)\ge0$.
\end{proof} 

\paragraph{Monotonicity of the occupied set.}
Finally, when applying Harris's inequality, we will make use of the following basic lemma. The lemma says that, when an additional key is inserted before a common suffix of insertions, although the additional key may change the locations chosen by later keys, it cannot make the final occupied set smaller.

\begin{lemma}[Occupied-set monotonicity]
  \label{lem:occupied-set-monotonicity}
  Let $A\subseteq B$ be two occupied sets, and fix a sequence of future keys together with their probe sequences. If this same sequence is inserted starting from $A$ and from $B$, then after every insertion the occupied set in the first table is contained in the occupied set in the second table.
\end{lemma}

\begin{proof}
It suffices to consider one insertion. Let $s$ be the first slot in its probe sequence that is empty relative to $A$. If $s\in B$, then the insertion into $A$ produces $A\cup\{s\}\subseteq B$, which is contained in the occupied set produced from $B$. If $s\notin B$, then every probe before $s$ lies in $A\subseteq B$, while $s$ is empty in both tables. Both insertions therefore choose $s$, so the inclusion is again preserved. Applying this one-step argument successively proves the lemma.
\end{proof}

%% file: poissonization.tex
In this section, we establish the key positive association property that we need in our analysis. Instead of working on a table with a fixed number of stored keys, we Poissonize the number of stored keys. We then show that, within this Poissonized table, the indicator random variables for which slots are occupied are guaranteed to be positively associated random variables.

\paragraph{Poissonizing the number of keys.}
Rather than considering a table with $m = (1 - \eps)n$ keys, and analyzing the cost of the next insertion, we will spend most of the paper considering a hash table with $N \sim \Poisson(m)$ keys, and analyzing the next insertion into that table. We refer to the former as the \defn{fixed-load model} and the latter as the \defn{Poissonized model}. We will also refer to the tables after $n$ and $N$ insertions as the \defn{fixed-load table} and the \defn{Poissonized table}, respectively.

One subtle point is that, in the Poissonized model, we need to consider the possibility that more than $n$ total insertions occur (i.e., the table is full). In that case, we consider subsequent insertions to be no-ops (they do not change the state of the table), and define their insertion time to be $2n$. 

The next lemma says that, to analyze the fixed-load model, it suffices to analyze the Poissonized model.
\begin{lemma}[Poissonization reduction]
  \label{lem:reduction-to-poissonized}
  Let $T_i$ denote the time to perform the $(i + 1)$-th insertion (where $T_i = 2n$ for $i \ge n$). Then, for any positive integer $m$, and for $N \sim \Poisson(m)$, we have
  \begin{align}
      \Pr[T_{m} >k] \le 2\Pr[T_{N} > k].
      \label{eq:poisson-comparison}
  \end{align}
\end{lemma}
\begin{proof}
Let $x_1, x_2, \ldots$ be distinct keys, and let $A_r$ denote the occupied set after inserting $x_1, \ldots, x_r$. Let $y \not\in \{x_i\}$ be an additional key, and let $\overline{T}_r$ be the time to insert $y$ into the table formed from $x_1, \ldots, x_r$. Then, by construction, $A_1 \subseteq A_2 \subseteq \cdots$, and $\overline{T}_1 \leq \overline{T}_2 \leq \cdots$. This means that, on the event $N \ge m$, we have $\overline{T}_m \le \overline{T}_N$. Note also $\overline{T}_m$ and $\overline{T}_N$ have the same distributions as $T_m$ and $T_N$, respectively.

Since $m$ is a positive integer and $N$ is Poisson with mean $m$, we have $\Pr[N\ge m]\ge1/2$ \cite[p.~130]{mitzenmacher2005probability}. We therefore have
\begin{align*}
    \Pr[\overline{T}_N >k]
    &\ge \Pr[N\ge m,\ \overline{T}_m>k]\\
    &=\Pr[N\ge m]\Pr[\overline{T}_m>k]\\
    &\ge\frac12 \Pr[\overline{T}_m > k]. \qedhere
\end{align*}
\end{proof}

\paragraph{Positive association in the Poissonized table.}
Let $N \sim \Poisson((1-\eps)n)$ be the number of keys inserted into the Poissonized table. For each $i\in[n]$, let $F_i\in\{0,1\}$ be the indicator random variable for whether slot $i$ is empty in the Poissonized table. We now argue that the $F_i$s are positively associated.

\begin{lemma}
  \label{lemma:vacancy-mean-and-association}
  For each $i\in[n]$, we have $\E\Bk*{F_i}\ge\eps$. Moreover, the variables $(F_i)_{i\in[n]}$ are positively associated.
\end{lemma}

It's worth noting that \cref{lemma:vacancy-mean-and-association} is true only because of Poissonization. Consider, for example, what happens if there are exactly $n - 1$ keys inserted into the table. Then there is exactly one empty slot, so the $F_i$s are actually \emph{negatively} associated. (Of course, we would be happy if the $F_i$s were always negatively associated, but that is also not the case.) Only after Poissonizing the number of keys can we conclude that the $F_i$s are positively associated.

In the discussion, below, we use 
\[
    Z\defeq\sum_{i\in[n]}F_i
\]
to denote the number of empty slots in the Poissonized table. Equivalently,
\begin{align}
    Z=(n-N)_+,
    \label{eq:empty-slots}
\end{align}
where $(\cdot)_+$ replaces negative values by $0$. 

We will make repeated use of the following basic fact, which follows by rotational symmetry: 
\begin{observation}[Translation invariance]
  \label{obs:translation-invariance}
  For every $c$, the shifted vector $(F_{i+c})_{i\in[n]}$ has the same distribution as $(F_i)_{i\in[n]}$, where slot indices are taken modulo $n$.
\end{observation}

With this setup in mind, we now show how to use Harris's inequality to prove \cref{lemma:vacancy-mean-and-association}.

\begin{proofof}{\cref{lemma:vacancy-mean-and-association}}
By \cref{obs:translation-invariance}, the means $\E\Bk*{F_i}$ are all equal, so each of them is $\E\Bk*{Z}/n$. Since $(n-N)_+\ge n-N$, \eqref{eq:empty-slots} gives
\[
    \E\Bk*{F_i}=\frac1n\E\Bk*{(n-N)_+}\ge\frac{n-\E\Bk*{N}}n=\eps.
\]

We prove positive association in two steps. First, we discretize the insertion times into finitely many layers and apply Harris's inequality, using the occupied-set monotonicity from \cref{lem:occupied-set-monotonicity}. We then let the number of layers tend to infinity, and argue that the resulting processes converges to the Poissonized model.

Fix an integer $q\ge1$. For every $t\in[q]$ and $i\in[n]$, let $Y_{t,i}$ be independent random variables with distribution $\Poisson((1-\eps)/q)$. Given the $\left(Y_{t,i}\right)_{t\in[q],i\in[n]}$ random variables, construct a quadratic probing hash table by inserting $Y_{t,i}$ keys with hash $i$ in lexicographic order of $(t,i)$, and let $F_i^{(q)}$ indicate whether slot $i$ is empty after all these insertions. 

We now argue that $(F_i^{(q)})_{i\in[n]}$ are positively associated. Increasing one coordinate $Y_{t,i}$ inserts additional keys at that point in the order. Immediately after these additional insertions, the augmented experiment has a superset of the occupied slots in the original experiment; the remaining keys and their probe sequences are the same in both experiments. By \cref{lem:occupied-set-monotonicity}, this inclusion persists through the remaining insertions. Thus every $F_i^{(q)}$ is coordinatewise nonincreasing in the independent family $(Y_{t,i})_{t\in[q],i\in[n]}$. \Cref{cor:monotone-association} therefore implies that $(F_i^{(q)})_{i\in[n]}$ is positively associated.

We next compare the above experiment, as the number $q$ of layers tends towards infinity, to the Poissonized table. For $t \in [q]$, let $Y_t=\sum_{i\in[n]}Y_{t,i}$ be the number of keys in layer $t$. The variables $Y_t$ are independent with distribution $\Poisson((1-\eps)n/q)$, and their sum is exactly $\Poisson((1-\eps)n)$. Conditional on the values $(Y_t)_{t\in[q]}$, the keys in every layer have independent uniform hashes. Moreover, for a Poisson random variable of mean $\lambda$, the probability of the random varaible being at least $2$ is at most $\lambda^2/2$. A union bound therefore gives
\[
    \Pr[\text{some layer contains at least two keys}]
    \le q\cdot\frac{((1-\eps)n/q)^2}{2}
    =O\bk*{\frac{n^2}{q}}.
\]
Equivalently, one may generate the layered experiment by first sampling the Poisson number of keys and then assigning each key an independent uniform layer and an independent uniform hash. If every layer contains at most one key, ordering the keys by their layers merely permutes independent uniform hashes, so the resulting insertion sequence has exactly the desired Poissonized distribution. The two experiments can therefore differ only on the event displayed above, whose probability tends to zero. Hence the layered insertion sequence converges in distribution to the Poissonized insertion sequence as $q\to\infty$. Since the vacancy vector takes values in the finite set $\{0,1\}^n$, expectations of bounded functions of that vector converge as well. The covariance inequalities defining positive association consequently pass to the limit, proving that $(F_i)_{i\in[n]}$ is positively associated.
\end{proofof}

\paragraph{Defining the vacancy covariances.}
By \cref{obs:translation-invariance}, the covariance $\Cov(F_i,F_{i+r})$ does not depend on $i$. We may therefore define the \defn{covariance function}
\[
    C(r)\defeq\Cov(F_i,F_{i+r}),
\]
where $i\in[n]$ is arbitrary. Since slot indices are taken modulo $n$, the value $C(r)$ depends only on $r$ modulo $n$. Moreover, because the $F_i$s are positively associated, the covariance is guaranteed to be nonnegative:

\begin{observation}
  \label{obs:covariance-nonnegative}
  For every integer $r$, we have $C(r)\ge0$.
\end{observation}

\begin{proof}
By construction, $C(r) = \Cov(f, g)$ where $f = F_i$ and $g = F_{i + r}$ are both nondecreasing functions of $(F_i)_{i\in[n]}$. Since $(F_i)_{i\in[n]}$ are positively associated (Lemma \ref{lemma:vacancy-mean-and-association}), it follows that $C(r)\ge0$.
\end{proof}

The fact that $C(r)$ is nonnegative will play a critical role in the next section, where it allows us to obtain a clean bound on the sum $\sum_{0\le a<b<k}C(b^2-a^2)$ of the covariances between slots that are probed in the first $k$ steps of an insertion. 

\section{Bounding the Vacancy Covariances along a Probe Sequence}
\label{sec:covariances}

Now consider what happens when we insert one more key into the Poissonized table (i.e., the $(N+1)$-st insertion). By \cref{obs:translation-invariance}, we may assume that the new key's hash is $0$, so that the first $k$ probes by the insertion are to slots $0^2,1^2,\ldots,(k-1)^2$. To reason about the number of free slots among these probes, it will be helpful to obtain a bound on
\begin{equation}
    \sum_{0\le a<b<k} \Cov(F_{a^2}, F_{b^2}) = \sum_{0\le a<b<k} C(b^2-a^2),
    \label{eq:sum-of-covariances-along-probe-sequence}
\end{equation}
which is the sum of the covariances over all distinct pairs from $F_{0^2}, F_{1^2}, \ldots, F_{(k-1)^2}$.

Define
\begin{align}
    R(k)\defeq\max_{1\le r<n}\abs*{\{(a,b):0\le a<b<k,\ b^2-a^2\equiv r\pmod n\}}
    \label{eq:def-R}
\end{align}
to be the largest number of times that any $r \in [n]$ appears as a difference between squares $b^2 - a^2 \pmod n$ for distinct $a < b \in [k]$. Equivalently, if we rewrite \eqref{eq:sum-of-covariances-along-probe-sequence} as a linear combination of $C(0), \ldots, C(k - 1)$, then $R(k)$ is the largest coefficient in this linear combination. The next proposition is one of the central technical observations of the paper: it says that the sum of covariances in \eqref{eq:sum-of-covariances-along-probe-sequence} is at most $R(k)$.

\begin{proposition}
  \label{prop:small-covariances-along-probe-sequence}
  Let $n$ be prime and let $1\le k\le n/2$. Then
  \[
      \sum_{0\le a<b<k}C(b^2-a^2)\le R(k).
  \]
\end{proposition}
Of course, whether Proposition \ref{prop:small-covariances-along-probe-sequence} is useful depends on whether $R(k)$ is small. Later in the paper, we will see that, when $k \le \sqrt{n}$, we have $R(k) = k^{o(1)}$. (This is to say that, to a first approximation, the distances $\{b^2-a^2 \pmod n: 0 \le a < b < k\}$ are mostly distinct.) The fact that $R(k)$ is so small is a one of the key properties that distinguishes quadratic probing from linear probing, and is the main reason that we are able to in our final analysis obtain a time bound of $\epsilon^{-(1 + o(1))}$ (which is better than the $\epsilon^{-2}$-style bound that holds for linear probing).

The main step in proving \cref{prop:small-covariances-along-probe-sequence} is to first obtain the following bound on the sum of all $C(r)$s.

\begin{lemma}
  \label{lemma:total-covariance}
  We have $\sum_{r=0}^{n-1}C(r)\le1$.
\end{lemma}

\begin{proof}
Recall from \eqref{eq:empty-slots} that the number of empty slots is $Z=(n-N)_+$. Expanding the variance of $Z$,
\[
\begin{aligned}
    \Var(Z)
    &=\sum_{i\in[n]}\sum_{j\in[n]}\Cov(F_i,F_j)\\
    &=n\sum_{r=0}^{n-1}C(r).
\end{aligned}
\]

For any random variable $X$, we have $\Var(X) \ge \Var(\min(X, n))$. Since $\Var(Z) = \Var(\min(N, n))$, it follows that $\Var(Z) \le \Var(N)$. 
Since $N$ is a Poisson random variable with mean and variance $(1-\eps)n$, this implies that $\Var(Z) \le n$. Plugging this into the identity above gives $\sum_{r=0}^{n-1}C(r) \le 1$.
\end{proof}

Note that, if we were using a fixed-load table rather than a Poissonized one, the argument above would give the even stronger-looking identity $\sum_{r=0}^{n-1}C(r)=0$, because $Z$ would be non-random and hence $\Var(Z)$ would be $0$. But that identity would be useless to us, since the only way that $\sum C(r)$ can be zero is if some of the terms are cancelling with others. What makes \cref{lemma:total-covariance} useful to us is that, thanks to Poissonization, the $C(r)$s are nonnegative.

Because the $C(r)$s are nonnegative, we can deduce from \cref{lemma:total-covariance} that $\sum_{r\in S}C(r)\le1$ for every subset $S\subseteq[n]$. This leads to the following proof of \cref{prop:small-covariances-along-probe-sequence}.

\begin{proofof}{\cref{prop:small-covariances-along-probe-sequence}}
Grouping the pairs according to the residue $r=b^2-a^2$ gives
\[
    \sum_{0\le a<b<k}C(b^2-a^2) =\sum_{r=1}^{n-1}\abs*{\{(a,b):0\le a<b<k,\ b^2-a^2\equiv r\pmod n\}} \cdot C(r),
\]
where the omission of the $r = 0$ term makes use of the fact that $b^2 - a^2 \not\equiv 0 \pmod n$ for all $0 \le a < b < k$ (Observation \ref{obs:distinct-probes}).
By the definition of $R(k)$, and the fact that $C(r)\ge0$, this is at most
$$R(k)\sum_{r=1}^{n-1}C(r),$$
which by \cref{lemma:total-covariance} is at most $R(k)$.
\end{proofof}

%% file: covariance-tail-bound.tex
In this section we convert the covariance estimate from Proposition~\ref{prop:small-covariances-along-probe-sequence} into a bound on the insertion time in the Poissonized table. The main result of the section is the following tail bound, which relates the insertion time to the quantity $R(k)$ defined in \eqref{eq:def-R}.

\begin{restatable}[Poissonized tail bound in terms of $R(k)$]{proposition}{proptailviaR}
  \label{prop:tail-via-R}
  Let $n$ be prime and let $1<k\le n/2$, and let $T$ denote the time to insert one more key into the Poissonized table. Then
  \begin{align}
      \Pr[T>k]=O\bk*{\frac{R(k)}{\eps^2k^2}}.
      \label{eq:tail-via-R}
  \end{align}
\end{restatable}

By \cref{obs:translation-invariance}, we may assume that the final key being inserted has hash $0$, so its first $k$ probes are the slots $0^2,1^2,\ldots,(k-1)^2$. The insertion incurs more than $k$ probes if and only if
$$F_{0^2}=F_{1^2}=\cdots=F_{(k-1)^2}=0.$$
The following lemma gives a general-purpose bound on the probability of a collection of positively associated indicator random variables being all $0$.
\begin{lemma}
  \label{lem:covariance-tail-bound}
  Let $X_1,\dots,X_n$ be positively associated random variables taking values in $\{0,1\}$. Let $S=\sum_{i=1}^{n}X_i$, let $\mu=\E\Bk*{S}>0$, and let $d=\sum_{1\le i<j\le n}\Cov(X_i,X_j)$. Then
  \[
      \Pr[S=0]\le e^{-\mu/8}+\frac{8d}{\mu^2}.
  \]
\end{lemma}

It is worth noting that Lemma \ref{lem:covariance-tail-bound} is one of two places in the analysis where we make critical use of positive association (the first being the proof of Proposition~\ref{prop:small-covariances-along-probe-sequence}). Without positive association, the natural substitute for Lemma \ref{lem:covariance-tail-bound} would be Chebyshev's inequality, which says that $\Pr[S=0]\le \Var(S)/k^2$. In the case where $S = \sum_{i = 1}^k F_{i^2}$, the variance $\Var(S)$ is dominated by $\sum_{i = 1}^k \Var(F_{i^2}) = \Theta(k \epsilon)$. Therefore, without positive association, the bound we would be left with is $\Pr[S=0]\le O(k \epsilon)/k^2 = O(\epsilon/k)$, which would be too weak to get \emph{any} non-trivial bound on insertion time as a function of $\epsilon^{-1}$. In contrast, Lemma \ref{lem:covariance-tail-bound} allows us to get a bound closer to $O(1/k^2)$ so long as $d$ is small. 

\begin{proof}[Proof of Lemma \ref{lem:covariance-tail-bound}]
The idea is to compare $S$ with the sum of independent Bernoulli variables  $X^{\perp}_1, \ldots, X^{\perp}_n$, where each $X^{\perp}_i$ has the same distribution as $X_i$, and to bound the cost of this comparison via Claim \ref{claim:convex-order-inequality} below. We can then analyze the sum of independent Bernoulli variables with a standard Chernoff bound.

Let $X^{\perp}_1,\dots,X^{\perp}_n$ be mutually independent random variables, independent of $(X_i)_{i\in[n]}$, such that $X_i^{\perp}$ has the same distribution as $X_i$. Write $S^{\perp}=\sum_{i=1}^{n}X_i^{\perp}$. We make use of the following claim.

\begin{claim}[Corollary 2 of \cite{boutsikasvaggelatou2002distance}, rephrased]
  \label{claim:convex-order-inequality}
  For every convex function $\varphi$ for which the expectations exist,
  \[
      \E\Bk*{\varphi(S^{\perp})}\le\E\Bk*{\varphi(S)}.
  \]
\end{claim}

Let $a=\mu/2$ and set
\[
    h(x)=\bk*{\max\{0,1-x/a\}}^2,
    \qquad
    g(x)=\frac{x^2}{a^2}-h(x).
\]
The definition is chosen so that $h$ detects values below $a$, while $g$ is convex. Indeed,
\[
    g(x)=
    \begin{cases}
      2x/a-1, & x\le a,\\
      x^2/a^2, & x\ge a.
    \end{cases}
\]
The two pieces have the same value $1$ and the same derivative $2/a$ at $x=a$, and the derivative of the quadratic piece is nondecreasing. Hence $g$ is convex. Applying Claim \ref{claim:convex-order-inequality} gives $\E\Bk*{g(S^{\perp})} \le\E\Bk*{g(S)}$. Substituting the definition of $g$ into this inequality and rearranging gives
\begin{equation}
    \E\Bk*{h(S)} \le \E\Bk*{h(S^{\perp})} + \frac{\E\Bk*{S^2-(S^{\perp})^2}}{a^2}.
    \label{eq:convex-order-inequality-inequality}
\end{equation}
Observe that
\begin{align*}
    \E\Bk*{S^2-(S^{\perp})^2}
    &=\Var(S)-\Var(S^{\perp}) \tag{since $\E\Bk*{S}=\E\Bk*{S^{\perp}}$}\\
    &=\bk*{\sum_{i=1}^{n}\Var(X_i)+2\sum_{1\le i<j\le n}\Cov(X_i,X_j)}
      -\sum_{i=1}^{n}\Var(X_i^{\perp}) \\
    &=2\sum_{1\le i<j\le n}\Cov(X_i,X_j) \tag{since $\Var(X_i)=\Var(X_i^{\perp})$} \\
    & =2d. 
\end{align*}
Substituting this into \eqref{eq:convex-order-inequality-inequality} gives
\begin{equation}
    \E\Bk*{h(S)} \le \E\Bk*{h(S^{\perp})} + \frac{2d}{a^2}.
    \label{eq:convex-order-inequality-inequality-2}
\end{equation}

Because $S$ and $S^{\perp}$ are nonnegative, $h(S)\ge\ind\nolimits_{\BK*{S=0}}$ and $h(S^{\perp})\le\ind\nolimits_{\BK*{S^{\perp}<a}}$. Combining these inequalities with \eqref{eq:convex-order-inequality-inequality-2} yields
\[
    \Pr[S=0]\le\Pr[S^{\perp}<a]+\frac{2d}{a^2}.
\]
Finally, $S^{\perp}$ is a sum of independent Bernoulli random variables with mean $\mu=2a$, so a Chernoff bound gives
\[
    \Pr[S^{\perp}<a]
    =\Pr\Bk*{S^{\perp}<\frac{\mu}{2}}
    \le e^{-\mu/8}.
\]
Substituting $a=\mu/2$ makes $2d/a^2=8d/\mu^2$ and proves the lemma.
\end{proof}

Combining this estimate with the covariance bound of \cref{prop:small-covariances-along-probe-sequence} gives Proposition \ref{prop:tail-via-R}.

\proptailviaR*

\begin{proof}
Since $(F_{a^2})_{a=0}^{k-1}$ are positively associated by \cref{lemma:vacancy-mean-and-association}, we can apply Lemma \ref{lem:covariance-tail-bound} to them to get
\begin{align*}
    \Pr[T>k]
    &=\Pr\Bk*{\sum_{a=0}^{k-1}F_{a^2}=0}\\
    &\le \exp\bk*{-\frac{\E\Bk*{\sum_{a=0}^{k-1}F_{a^2}}}{8}}
      +\frac{8\sum_{0\le a<b<k}\Cov(F_{a^2},F_{b^2})}
      {\bk*{\E\Bk*{\sum_{a=0}^{k-1}F_{a^2}}}^2}\\
    &\le e^{-\eps k/8}+\frac{8}{\eps^2k^2}\sum_{0\le a<b<k}C(b^2-a^2). \tag{since $\E\Bk*{\sum_{a=0}^{k-1}F_{a^2}}\ge\eps k$}
\end{align*}
By Proposition \ref{prop:small-covariances-along-probe-sequence}, the covariance sum is at most $R(k)$, so
\[
    \Pr[T>k]\le e^{-\eps k/8}+\frac{8R(k)}{\eps^2k^2}.
\]
It remains to eliminate the exponential term, which we do by observing that it never exceeds the order of the second term. Indeed, $e^{-x/8}=O(x^{-2})$ for $x>0$, since the exponential decays faster than any power; taking $x=\eps k$ gives
\[
    e^{-\eps k/8}=O\bk*{\frac1{\eps^2k^2}} = O\bk*{\frac{R(k)}{\eps^2k^2}}. \qedhere
\]
\end{proof}

%% file: tail-bounds.tex
Now, we return our attention to the fixed-load setting. Throughout the section, we consider a fixed-load table with $m = (1 - \epsilon)n$ keys, and we derive a tail bound on the time $T$ to perform one additional insertion.

\begin{proposition}
  \label{prop:tail-bound-for-T}
  Suppose $n$ is a sufficiently large prime and $1/\sqrt{\log n}<\eps\le1/2$. Let $T$ be the time to perform one additional insertion in a fixed-load table with $m = (1 - \epsilon)n$ keys. Then
  \begin{subnumcases}{\Pr[T>k]\le}
    \dfrac{k^{o(1)}}{\eps^2k^2}, & for $\eps^{-1}<k\le n^{1/2}$,\label{case:divisor}\\
    O\bk*{\dfrac1{\eps^2n}}, & for $k=\ceil*{n^{0.76}}$,\label{case:hyperbola}\\
    \dfrac{3\ln\eps^{-1}}{k}, & for $\eps n\le k\le n/2$.\label{case:long-tail}
  \end{subnumcases}
\end{proposition}

These three bounds do not cover every $k$, but they will nonetheless suffice for our time analysis in the next section. Since $\Pr[T>k]$ is nonincreasing in $k$, a bound at any single $k$ also applies to every larger $k$; this is what lets us get away with the single value $k=\ceil*{n^{0.76}}$ in case \eqref{case:hyperbola}, and it is also how we will cover the gap between $n^{1/2}$ and $n^{0.76}$ when we sum the tail in \cref{sec:expected-time}.

We obtain cases \eqref{case:divisor} and \eqref{case:hyperbola} by first bounding the quantity $R(k)$, plugging that into \cref{prop:tail-via-R}, and then applying the Poissonization comparison in Lemma \ref{lem:reduction-to-poissonized}. The two cases correspond to two different arithmetic estimates for $R(k)$, and they are handled in \cref{subsec:bounds-for-R(k)}.

It is worth noting that the bounds for $k < \eps n$ are already sufficient to obtain a non-trivial overall bound on insertion time, namely a bound of $\E[T] = O(\epsilon^{-2})$. If, however, we want a bound of $O(\epsilon^{-(1 + o(1))})$, then we must also establish a nontrivial tail bound for $k\ge\eps n$. This final tail bound is established via an entirely different ad-hoc argument. In this extreme parameter regime, we can condition on an arbitrary fixed table state and argue that, no matter what that state is, the randomness of the hash of the next insertion is enough to obtain the bound $3\ln\eps^{-1}/k$. This case is handled in \cref{subsec:long-tail-bound-via-covering-lemma}.

We remark that the arguments in this section are the only parts of the proof that make use of facts specific to quadratic probing (as opposed to other fixed-offset probing schemes).
They also reveal the main reason that quadratic probing is able to avoid the clustering effects of linear probing, which force linear probing to incur $\Omega(\epsilon^{-2})$-time operations: it is because, unlike linear probing where the quantity $R(k)$ grows as $\Theta(k)$, quadratic probing has its probes spread far enough out that $R(k)$ grows (at least initially) at a much slower at a rate of $k^{o(1)}$.

\subsection{Bounds for \texorpdfstring{$R(k)$}{the Multiplicity}, and Proofs of \texorpdfstring{\eqref{case:divisor} and \eqref{case:hyperbola}}{(8a) and (8b}}
\label{subsec:bounds-for-R(k)}
Recall from \eqref{eq:def-R} that $R(k)$ is the largest number of pairs among the first $k$ probe indices that produce the same square difference modulo $n$. That is, $R(k)$ is the largest number, over all residues $r$ with $1\le r<n$, of pairs $(a,b)$ with $0\le a<b<k$ satisfying
\begin{align}
    b^2-a^2\equiv r\pmod n.
    \label{eq:square-difference-congruence}
\end{align}
Both estimates below start from the factorization $b^2-a^2=(b-a)(b+a)$, which turns counting solutions of \eqref{eq:square-difference-congruence} into counting factorizations. Up to $k\le n^{1/2}$, both sides of \eqref{eq:square-difference-congruence} are smaller than $n$, so the congruence is an equality of integers and a single divisor bound suffices. Near $n^{0.76}$, the difference $b^2-a^2$ can exceed $n$, so many integers in the congruence class of $r$ become available and the divisor bound degrades; there we instead use a modular-hyperbola estimate for the solutions of \eqref{eq:square-difference-congruence}.

\paragraph{The divisor bound for $k\le n^{1/2}$.}
Let $\tau(i)$ denote the number of positive divisors of $i$. If $k\le n^{1/2}$, then
\[
    0<b^2-a^2<k^2\le n.
\]
Hence \eqref{eq:square-difference-congruence} is an equality of integers: $b^2-a^2=r$. For a fixed value $r$, every such pair gives the factorization
\[
    r=b^2-a^2=(b-a)(b+a).
\]
The two factors $u=b-a$ and $v=b+a$ determine $a=(v-u)/2$ and $b=(v+u)/2$. Moreover, given $r$, the factor $u$ determines $v = r / u$. The number of options for $(a, b)$ for a given $r$ is therefore at most the number of divisors $u$ of $r$. Taking a maximum over the options for $r$ gives 
\[
    R(k)\le\max_{1\le i<k^2}\tau(i)=k^{o(1)},
\]
where the last estimate is the standard maximal-order bound for the divisor function~\cite{hardywright}. Substituting this estimate into \eqref{eq:tail-via-R} and then applying the Poissonization comparison \eqref{eq:poisson-comparison} shows that, in the fixed-load model,
\[
    \Pr[T>k]\le\frac{k^{o(1)}}{\eps^2k^2}
    \qquad\text{for }\eps^{-1}<k\le n^{1/2},
\]
which is case \eqref{case:divisor} of \cref{prop:tail-bound-for-T}.

\paragraph{The modular-hyperbola bound at $k=\ceil*{n^{0.76}}$.}
The divisor bound above relies on $b^2-a^2$ being smaller than $n$, so it gives nothing once $k$ exceeds $n^{1/2}$. For larger $k$ we use the following modular-hyperbola estimate instead.

\begin{lemma}
  \label{lem:modular-multiplicity}
  Let $n$ be an odd prime and let $1 < k\le n/2$. Then
  \[
      R(k)=O\bk*{\frac{k^2}{n}+n^{1/2+o(1)}}.
  \]
\end{lemma}

\begin{proof}
We use the following standard estimate for points on a modular hyperbola.

\begin{lemma}[Theorem 13 of \cite{shparlinski2012modular}, rephrased]
  \label{lem:shparlinski}
  Let $n$ be prime and let $r\not\equiv0\pmod n$. For integers $U,V\in[0,n)$ and $1\le X,Y<n$,
  \[
      \abs*{\{(i,j):1\le i\le X,\ 1\le j\le Y,\ (U+i)(V+j)\equiv r\pmod n\}}=O\bk*{\frac{XY}{n}+n^{1/2+o(1)}}.
  \]
\end{lemma}

Fix a residue $r$ with $1\le r<n$. For every pair $(a,b)$ satisfying \eqref{eq:square-difference-congruence}, set
\[
    u=a+b,
    \qquad
    v=b-a.
\]
Then $1\le u<2k$, $1\le v<k$, and
\[
    uv=(a+b)(b-a)=b^2-a^2\equiv r\pmod n.
\]
The map $(a,b)\mapsto(u,v)$ is injective because $a=(u-v)/2$ and $b=(u+v)/2$. We may therefore discard the parity and triangular constraints on $(u,v)$ and count all pairs in the containing rectangle. Since $k\le n/2$ implies $2k-1<n$, \cref{lem:shparlinski}, with $U=V=0$, $X=2k-1$, and $Y=k-1$, gives
\[
\begin{aligned}
    &\abs*{\{(a,b):0\le a<b<k,\ b^2-a^2\equiv r\pmod n\}}\\
    \le{}&\abs*{\{(u,v):1\le u<2k,\ 1\le v<k,\ uv\equiv r\pmod n\}}\\
    ={}&O\bk*{\frac{k^2}{n}+n^{1/2+o(1)}}.
\end{aligned}
\]
Taking the maximum over $r$ proves the lemma.
\end{proof}

Now take $k=\ceil*{n^{0.76}}$, so that $k^2/n\ge n^{0.52}$. The ratio of the error term in \cref{lem:modular-multiplicity} to $k^2/n$ is then at most
\[
    \frac{n^{1/2+o(1)}}{n^{0.52}}=o(1),
\]
so $R(k)=O(k^2/n)$. Substituting into \eqref{eq:tail-via-R} and applying \eqref{eq:poisson-comparison} gives
\[
    \Pr[T>k]=O\bk*{\frac{k^2/n}{\eps^2k^2}}
    =O\bk*{\frac1{\eps^2n}},
\]
which is case \eqref{case:hyperbola} of \cref{prop:tail-bound-for-T}.

\subsection{Proof of \texorpdfstring{\eqref{case:long-tail}}{(8c)} via a Covering Lemma}
\label{subsec:long-tail-bound-via-covering-lemma}
The covariance method used to prove \eqref{case:divisor} and \eqref{case:hyperbola} fixes the new key's hash and averages over the random table state. To prove the final case \eqref{case:long-tail}, we reverse the roles of the randomness: we condition on an arbitrary occupied set and use only that the hash of the next key is uniformly random. The Covering Lemma below shows that, no matter what set of $(1 - \epsilon)n$ slots are occupied, there are at most a very small number of options for the hash $h(x)$ of the next insertion such that the insertion takes a very long time. 

Throughout this subsection we add sets of slots, always modulo the table size: for $A,B\subseteq[n]$, their set sum is defined as
\[
    A+B\defeq\{a+b\bmod n:a\in A,\ b\in B\}.
\]

\begin{lemma}[The Covering Lemma]
  \label{lem:covering}
  Let $n$ be a sufficiently large odd prime, let $0<\eps\le1/2$ with $\eps>1/\sqrt{\log n}$, and let $\eps n\le k\le n/2$. Set $L_{k,n}=\ceil*{2(n/k)\ln\eps^{-1}}$. For every set $\Lambda\subseteq[n]$ of occupied slots with $\abs{\Lambda}=(1-\eps)n$, there are fewer than $L_{k,n}$ hashes $j\in[n]$ such that
  \[
      \{0^2,1^2,\ldots,(k-1)^2\}+j\subseteq\Lambda.
  \]
\end{lemma}

\begin{proof}
Write
\[
    Q=\{i^2\bmod n:0\le i<k\}.
\]
We wish to show that there are fewer than $L_{k,n}$ hashes $j\in[n]$ such that $Q+j\subseteq\Lambda$. It suffices to prove the contrapositive: for every set $J\subseteq[n]$ of $L_{k,n}$ distinct shifts, the union $Q+J=\bigcup_{j\in J}(Q+j)$ has more than $(1-\eps)n$ elements. 

We estimate the size of the union $Q + J$ using the inclusion-exclusion principle, using the following intersection bound as a building block.\footnote{To convert Lemma 3.1 of \cite{makzaharescu2011poisson} to our setting, we may use their lemma as follows. Take their prime to be $n$ and their hyperelliptic curve $y^2=f(x)$ to be $y^2=x$, so that $f$ is not a square and $d=\deg f=1$. Take their $y$-coordinate interval to be $\mathcal{I}=[0,k)$, so that their set of admissible $x$-coordinates is $\{y^2:0\le y<k\}=Q$. Their Section~3 carries the standing assumption $\mathcal{I}\subseteq[0,(n-1)/2]$, which our hypothesis $k\le n/2$ supplies. Take their auxiliary interval $\mathcal{J}$ to be all of $[0,n)$, which makes the condition on their rational function $g$ vacuous ($g$ is a constant); and use Remark~3.2 to drop the requirement that $\{1,g_1,\ldots,g_r\}$ be linearly independent. Finally, write $J=\{j_0,j_1,\ldots,j_r\}$ and translate $j_0$ to $0$, so that their shift set is $\mathcal{H}=\{j_0-j_1,\ldots,j_0-j_r\}$ and their parameter is $r=\abs{J}-1$. Their main term $\abs{\mathcal{I}}^{r+1}\abs{\mathcal{J}}^{r}/n^{2r}$ is then $k^{\abs{J}}/n^{\abs{J}-1}$, and their error term $2^{3r+2}d(2^rd-1)\sqrt n\log^{2r+2}n+O(\sqrt n\log^{2r+1}n)$ is $\sqrt n\exp(O(\abs{J}\log\log n))$.}

\begin{lemma}[Lemma 3.1 of \cite{makzaharescu2011poisson}, rephrased]
  \label{lem:mak}
  Let $n$ be an odd prime, let $1\le k\le n/2$, and let $J\subseteq[n]$ satisfy $\abs{J}=o(\log n/\log\log n)$. Then
  \[
      \abs*{\bigcap_{j\in J}(Q+j)}=\frac{k^{\abs{J}}}{n^{\abs{J}-1}}\pm O\bk*{\sqrt n\exp\bk*{O(\abs{J}\log\log n)}}.
  \]
\end{lemma}

The assumptions $k\ge\eps n$ and $\eps>1/\sqrt{\log n}$ imply
\begin{equation}
    L_{k,n}  =\ceil*{2(n/k)\ln\eps^{-1}} \le2\eps^{-1}\ln\eps^{-1}+1
    =O\bk*{\sqrt{\log n}\,\log\log n}
    =o\bk*{\frac{\log n}{\log\log n}}.
    \label{eq:L-bound}
\end{equation}
Consequently, \cref{lem:mak} applies to every nonempty subset $J'\subseteq J$. For each such $J'$, write
\[
    \abs*{\bigcap_{j\in J'}(Q+j)}
    =\frac{k^{\abs{J'}}}{n^{\abs{J'}-1}}+E_{J'},
    \qquad
    \abs{E_{J'}}\le\sqrt n\exp\bk*{O(\abs{J'}\log\log n)}.
\]
The inclusion-exclusion principle gives
\[
\begin{aligned}
    \abs{Q+J}
    &=\sum_{\emptyset\ne J'\subseteq J}(-1)^{\abs{J'}-1}
      \abs*{\bigcap_{j\in J'}(Q+j)}\\
    &\ge\underbrace{\sum_{\emptyset\ne J'\subseteq J}(-1)^{\abs{J'}-1}z
      \frac{k^{\abs{J'}}}{n^{\abs{J'}-1}}}_{\text{main term}}
      -\underbrace{\sum_{\emptyset\ne J'\subseteq J}\abs{E_{J'}}}_{\text{error}}.
\end{aligned}
\]
We bound the main term from below and the error from above, and then combine the two.

We evaluate the main term by grouping subsets according to their size. For each $s$, there are $\binom{\abs{J}}s$ subsets $J'$ of size $s$, and $k^s/n^{s-1}=n(k/n)^s$. Therefore
\[
\begin{aligned}
    \sum_{\emptyset\ne J'\subseteq J}(-1)^{\abs{J'}-1}
      \frac{k^{\abs{J'}}}{n^{\abs{J'}-1}}
    &=n\sum_{s=1}^{\abs{J}}(-1)^{s-1}\binom{\abs{J}}s\bk*{\frac kn}^s\\
    &=n\bk*{1-\bk*{1-\frac kn}^{\abs{J}}},
\end{aligned}
\]
where the last equality is the binomial identity $\sum_{s=0}^{\abs{J}}(-1)^s\binom{\abs{J}}s x^s=(1-x)^{\abs{J}}$, with the $s=0$ term moved to the other side. By the definition of $L_{k,n}$ and the inequality $1-x\le e^{-x}$,
\[
    \bk*{1-\frac kn}^{\abs{J}}
    \le e^{-k\abs{J}/n}
    \le e^{-2\ln\eps^{-1}}
    =\eps^2,
\]
so the main term is at least $(1-\eps^2)n$.

Next we bound the error term. Every nonempty subset has size at most $\abs{J}$, and there are fewer than $2^{\abs{J}}$ such subsets, so
\[
\begin{aligned}
    \sum_{\emptyset\ne J'\subseteq J}\abs{E_{J'}}
    &\le\sqrt n\,2^{\abs{J}}\exp\bk*{O(\abs{J}\log\log n)}\\
    &=\sqrt n\exp\bk*{O(\abs{J}\log\log n)}.
\end{aligned}
\]
By \eqref{eq:L-bound},
\[
    \abs{J}\log\log n
    =o(\log n),
\]
so the error is $n^{1/2+o(1)} = \epsilon n^{1 - \Omega(1)}$, where the final equality uses $\epsilon \ge 1/\sqrt{\log n}$.

Combining the main and error terms,
\[
    \abs{Q+J}\ge(1-\eps^2)n-\epsilon n^{1 - \Omega(1)},
\]
which, for sufficiently large $n$, is strictly greater than $(1-\eps)n$.
\end{proof}

Condition on any fixed occupied set $\Lambda$ of size $(1-\eps)n$. The event $T>k$ means that the uniformly random hash of the next key is one of the bad shifts counted in \cref{lem:covering}. Therefore
\begin{align}
    \Pr[T>k\mid\Lambda]
    &\le\frac{L_{k,n}}n\notag\\
    &\le\frac{2\ln\eps^{-1}}k+\frac1n\notag\\
    &\le\frac{3\ln\eps^{-1}}k.
    \label{eq:covering-tail}
\end{align}
For the last line, $k\le n/2$ gives $1/n\le1/(2k)$, and $\eps\le1/2$ gives $\ln\eps^{-1}\ge\ln2>1/2$, so $1/(2k)\le(\ln\eps^{-1})/k$. Averaging over the fixed-load table state preserves the same bound. This proves case \eqref{case:long-tail} of \cref{prop:tail-bound-for-T} and completes its proof.

%% file: expected-time.tex
We now assemble the tail bounds of \cref{sec:tail-bounds} into a bound on the expected insertion time, allowing us to complete the proof of \cref{thm:main}.

\thmmain*

\begin{proof}
We may assume without loss of generality that $\eps\le1/2$ since, by \cref{lem:occupied-set-monotonicity}, the expected insertion time is nonincreasing in $\eps$; the $O(1)$ bound for $\eps=1/2$ therefore implies an $O(1)$ bound for every larger $\eps$.

Let $T$ be the insertion time. Every insertion finishes within $2n$ probes, so
\[
    \E\Bk*{T}=\sum_{k=0}^{2n-1}\Pr[T>k],
\]
and we bound these terms range by range, using throughout that $\Pr[T>k]$ is nonincreasing in $k$.

Each of the terms with $k\le\eps^{-1}$ is at most $1$, so
\begin{equation}
    \sum_{k\le\eps^{-1}}\Pr[T>k] \le \eps^{-1}+1 = O(\eps^{-1}).
    \label{eq:small-k}
\end{equation}
By case \eqref{case:divisor} and the fact that $\sum_{k>\eps^{-1}}k^{-2+o(1)}=\eps^{1-o(1)}$, we have
\begin{equation}
    \sum_{\eps^{-1}<k\le n^{1/2}}\Pr[T>k]
    \le \sum_{\eps^{-1}<k\le n^{1/2}}\frac{k^{o(1)}}{\eps^2k^2}
    =\frac{\eps^{1-o(1)}}{\eps^2}
    =\eps^{-(1+o(1))}.
    \label{eq:divisor-range}
\end{equation}
Note that this is the only range that contributes more than $O(\eps^{-1})$.

For $n^{1/2}<k<\ceil*{n^{0.76}}$, we can apply case \eqref{case:divisor} at $k=\floor{n^{1/2}}$ to bound
\[
    \Pr[T>k] \le \Pr\Bk*{T>\floor{n^{1/2}}} \le \frac{n^{o(1)}}{\eps^2n}.
\]
There are fewer than $n^{0.76}$ such terms, so they contribute at most
\begin{equation}
    \sum_{n^{1/2}<k<\ceil*{n^{0.76}}}\Pr[T>k]
    \le \frac{n^{0.76+o(1)}}{\eps^2n}
    =o(1).
    \label{eq:middle-range}
\end{equation}
Likewise, case \eqref{case:hyperbola} bounds $\Pr[T>k]$ by $O\bk*{1/(\eps^2n)}$ for every $k\ge\ceil*{n^{0.76}}$, so
\begin{equation}
    \sum_{\ceil*{n^{0.76}}\le k<\eps n}\Pr[T>k]
    =O\bk*{\frac{\eps n}{\eps^2n}}
    =O(\eps^{-1}).
    \label{eq:hyperbola-range}
\end{equation}

Finally, case \eqref{case:long-tail} applies to every $k$ with $\eps n\le k\le n/2$, and the harmonic sum over that range satisfies $\sum_{\eps n\le k\le n/2}1/k\le1+\ln\eps^{-1}$. Hence
\begin{equation}
    \sum_{\eps n\le k\le n/2}\Pr[T>k]
    \le\sum_{\eps n\le k\le n/2}\frac{3\ln\eps^{-1}}{k}
    \le3\ln\eps^{-1}\bk*{1+\ln\eps^{-1}}
    =O(\eps^{-1}).
    \label{eq:long-tail-range}
\end{equation}
Applying that same case at $k=\floor{n/2}$, we can bound the remaining $2n$ terms by
\begin{equation}
    \sum_{n/2<k<2n}\Pr[T>k]
    \le \sum_{n/2<k<2n} \Pr\Bk*{T>\floor{n/2}}
    \le  \sum_{n/2<k<2n} O\bk*{\ln\eps^{-1}/n}
    =O(\eps^{-1}).
    \label{eq:completion-range}
\end{equation}

Summing \eqref{eq:small-k}--\eqref{eq:completion-range} gives $\E\Bk*{T}=\eps^{-(1+o(1))}$.
\end{proof}